\documentclass[11pt]{article}
\usepackage[margin=1in]{geometry}
\usepackage[final]{microtype}
\usepackage{epsfig}
\usepackage{graphics}
\usepackage{latexsym}
\usepackage{amsmath}
\usepackage{amsfonts}
\usepackage{amssymb}
\usepackage{mathrsfs}
\usepackage[dvipsnames]{xcolor}
\usepackage{amsthm}
\usepackage{xspace}
\usepackage{epstopdf}
\usepackage{float}
\usepackage{hyperref}
\usepackage{pgfplots}
\pgfplotsset{compat=1.18}
\usepackage{longtable}

\numberwithin{equation}{section}
\usepackage[ruled,vlined]{algorithm2e}
\SetArgSty{textrm}

\usepackage[english]{babel}
\usepackage[nottoc]{tocbibind}

\usepackage{caption}
\usepackage{subcaption}

\usepackage{pifont}
\usepackage{booktabs}
\usepackage{multirow}

\theoremstyle{plain}
\newtheorem{theorem}{Theorem}

\newtheorem{corollary}{Corollary}

\newtheorem{lemma}{Lemma}

\theoremstyle{definition}

\theoremstyle{remark}

\makeatletter
\@addtoreset{equation}{section}
\def\section{\@startsection {section}{1}{\z@}{-3.5ex plus -1ex minus
 -.2ex}{2.3ex plus .2ex}{\large\bf}}
\makeatother

\def\bfm#1{\mbox{\boldmath$#1$}}

\def\0{\bfm 0}

\DeclareMathAlphabet{\mathpzc}{OT1}{pzc}{m}{it}
\newcommand{\argmin}{\operatorname*{arg\,min}}
\newcommand{\MS}{\operatorname{MS}}
\newcommand{\OPT}{\operatorname{OPT}}

\newcounter{my}

\newcounter{my2}

\newcounter{my3}

\newcounter{my4}

\newcounter{my5}

\newcounter{my6}

\allowdisplaybreaks

\begin{document}

\title{Consistency--Robustness Tradeoffs for\\
 Strategyproof Scheduling with Predictions}
\author{Hau Chan$^{1}$\quad Jianan Lin$^{2}$\quad Chenhao Wang $^{3,4}$\\[0.75em]
$1$ University of Nebraska-Lincoln\\
$2$ Rensselaer Polytechnic Institute\\
$3$ Beijing Normal University-Zhuhai\\
$4$ Beijing Normal-Hong Kong Baptist University
}
\date{}
\maketitle

\begin{abstract}
We study strategyproof scheduling on \(n\) unrelated machines with predictions.
Each machine is controlled by an agent with privately known processing times, while the mechanism receives a public prediction of the processing-time matrix before the agents report.
The objective is to minimize the makespan subject to strategyproofness.
We measure performance by consistency, the approximation guarantee for correct predictions, and robustness, the worst-case guarantee for arbitrary predictions.

We introduce \textsc{EdgeSkip}, a deterministic strategyproof member of the class of job-wise weighted mechanisms.
Such mechanisms allocate each job independently using prediction-dependent weights.
Using a polynomial-time \(2\)-approximate reference schedule computed from the prediction and a standard tradeoff parameter, \textsc{EdgeSkip} is \(4\)-consistent and \((2n-2)\)-robust, improving on the \((6,2n)\) guarantee of Balkanski, Gkatzelis, and Tan.
Without computational restrictions, \textsc{EdgeSkip} with an optimal reference schedule is \(C\)-consistent and \(\max\{n,(n-1)C/(C-1)\}\)-robust for every \(C>1\).
We prove a matching information-theoretic lower bound for all job-wise weighted mechanisms, thereby determining the exact consistency--robustness tradeoff for this class.
For arbitrary deterministic strategyproof mechanisms, we establish the robustness lower bound \(\max\{n,C/(C-1)\}\) for every \(C>1\).
We also study an error-tolerant variant and improve upon prior guarantees.
Experiments demonstrate that our mechanisms achieve good empirical performance.
\end{abstract}

\section{Introduction}\label{sec:introduction}

Scheduling on unrelated machines is a fundamental problem in approximation algorithms and algorithmic mechanism design \cite{lenstra1990approximation,nisan1999algorithmic}.
There are \(m\) jobs and \(n\) machines, each controlled by an agent, and agent \(i\) incurs processing time \(p_{ij}\) if job \(j\) is assigned to its machine.
The load of a machine is the total processing time of the jobs assigned to it, and the makespan is the maximum load over all machines.
The goal is to find an assignment of jobs to machines that minimizes the makespan.
When the processing times are public, a polynomial-time algorithm can compute an assignment with makespan at most twice the optimum \cite{lenstra1990approximation}.

In the strategic version, the processing-time vector of each agent is private information, and an agent may benefit from misreporting it.
A mechanism therefore chooses both an assignment and payments.
The utility of an agent is its payment minus the total true processing time of the jobs assigned to its machine.
A mechanism is strategyproof, or dominant-strategy incentive compatible, if truthful reporting maximizes every agent's utility regardless of the reports of the other agents.

Nisan and Ronen initiated the study of this problem and designed a deterministic strategyproof mechanism with approximation ratio \(n\) \cite{nisan1999algorithmic}.
They conjectured that no deterministic strategyproof mechanism can achieve a better approximation ratio, even without computational restrictions.
After a long sequence of studies establishing increasingly strong lower bounds on the approximation ratio, this conjecture was proved by Christodoulou, Koutsoupias, and Kov{\'a}cs \cite{christodoulou2023proof}.
Thus, even without computational restrictions, deterministic strategyproof mechanisms cannot achieve an approximation ratio better than \(n\), while the classical optimization version of the problem admits a polynomial-time \(2\)-approximation algorithm.
This sharp gap makes it natural to ask whether useful information about the problem instance can improve the approximation guarantee while maintaining strategyproofness.

\paragraph{Learning-Augmented Strategyproof Scheduling.}
The learning-augmented framework uses predictions to improve performance on well-predicted instances while retaining guarantees when predictions are inaccurate \cite{purohit2018improving,lykouris2021competitive,mitzenmacher2022algorithms}.
It has also been applied to strategic settings, including social choice, facility location, and general mechanism-design problems \cite{agrawal2022learning,xu2022mechanism}.
For strategyproof scheduling, the mechanism receives a public prediction \(\widehat{\mathbf p}\) of the true processing-time matrix \(\mathbf p=(p_{ij})\) before the agents report their types \cite{balkanski2023strategyproof}.
The prediction is supplied externally, is not controlled by the agents, and may be inaccurate.
Such a mechanism is evaluated by consistency and robustness.
Its \emph{consistency} is the approximation ratio when the prediction is correct, while its \emph{robustness} is the worst-case approximation ratio when the prediction is arbitrary.
The goal is to design a mechanism that is strategyproof for every fixed prediction and achieves the best possible consistency--robustness tradeoff.

The two extreme ways of using the prediction do not yield a satisfactory tradeoff.
Fixing an optimal schedule computed from the predicted matrix \(\widehat{\mathbf p}\), independently of the reports, gives perfect consistency and is strategyproof with report-independent payments, but its robustness is unbounded.
Conversely, ignoring the prediction and applying the mechanism of Nisan and Ronen achieves the optimal robustness factor of \(n\), but its consistency is also \(n\).

Balkanski, Gkatzelis, and Tan \cite{balkanski2023strategyproof} introduced \textsc{ScaledGreedy}, the first strategyproof mechanism with constant consistency and \(O(n)\) robustness for this setting.
Starting from a reference schedule for the predicted matrix, \textsc{ScaledGreedy} greedily moves jobs in decreasing order of predicted speedup and closes a machine once its incoming predicted load reaches a prescribed threshold. The resulting target schedule determines report-independent job-wise weights, and each job is then assigned independently to a machine minimizing its weighted reported processing time.
Given a reference schedule whose makespan is within a factor \(\kappa\) of the predicted optimum and a tradeoff parameter \(\gamma>0\), the mechanism is \((2+\gamma)\kappa\)-consistent and \(n(1+1/\gamma)\)-robust.
With the polynomial-time choices \(\kappa=2\) and \(\gamma=1\), these bounds become \(6\) and \(2n\).
They also introduced \textsc{ErrorTolerantScaledGreedy}, a tolerance-parameterized variant that biases the weighted allocation toward the prediction-derived target schedule and provides an error-dependent guarantee.

Both mechanisms belong to a natural class that we call \emph{job-wise weighted} mechanisms.
The prediction fixes a positive weight for every machine--job pair, and each job is assigned independently to the machine with the minimum weighted reported processing time.
Because the weights are independent of the reports, these allocation rules admit strategyproof payments.
The preceding results leave two central questions.
First, can one improve the polynomial-time guarantees and characterize the optimal consistency--robustness tradeoff within this weighted class?
Second, what lower bounds remain unavoidable for arbitrary deterministic strategyproof mechanisms, without restrictions on running time or allocation structure?

\subsection{Our Contributions}

We improve the polynomial-time guarantees, characterize the weighted frontier without computational restrictions, and make progress on the second question.
As usual, we describe allocation rules, while the strategyproof payments are guaranteed by their incentive properties and are made explicit in the preliminaries.

\begin{itemize}

    \item \textbf{Improved upper bounds (Section~\ref{sec:upper-bound}).}
    We introduce \textsc{EdgeSkip}, a prediction-dependent job-wise weighted mechanism.
    Starting from a \(\kappa\)-approximate reference schedule computed from the prediction, it constructs a target schedule by considering predicted speedups while skipping every proposed reassignment that would violate an incoming-load capacity.
    For every \(C>\kappa\), \textsc{EdgeSkip} can be parameterized to be strategyproof and \(C\)-consistent. Its robustness guarantee, denoted by \(R_{\kappa}(C)\), is the piecewise function stated in Theorem~\ref{thm:edge-skip}.
    The robustness analysis reduces the load on each output machine to the weight of a partial matching and uses a dependency forest to charge rejected reassignments without double counting.

    The general guarantee has two notable consequences.

    \noindent(1) With a polynomial-time \(2\)-approximate reference schedule, \textsc{EdgeSkip} can be parameterized to be \(4\)-consistent and \((2n-2)\)-robust, improving the \((6,2n)\) consistency--robustness guarantee of \textsc{ScaledGreedy}~\cite{balkanski2023strategyproof} and the \((4,2.5n)\) guarantee of Cole, Gupta, and Jangir~\cite{cole2025parsimonious}.
    \par\noindent(2) With an optimal reference schedule, \textsc{EdgeSkip} is \(C\)-consistent with robustness at most
    \(
        \max\left\{n,\frac{(n-1)C}{C-1}\right\}
    \)
    for every \(C>1\).

    We also introduce the error-tolerant variant \textsc{TolerantEdgeSkip}.
    For every \(C>\kappa\) and tolerance \(\bar\eta\ge1\), its approximation ratio is at most \(C\eta^2\) whenever the multiplicative prediction error \(\eta\) is at most \(\bar\eta\), and at most
    \(
        1+\bar\eta^2\bigl(R_{\kappa}(C)-1\bigr)
    \)
    for arbitrary predictions.

    \item \textbf{Lower bounds (Section~\ref{sec:lower-bounds}).} We prove that the \(C\)-consistency guarantee for any job-wise weighted mechanism requires robustness at least
    \(
        R\ge\max\left\{n,\frac{(n-1)C}{C-1}\right\},
    \)
    for every \(C>1\). This exactly matches \textsc{EdgeSkip} with an optimal reference schedule and determines the consistency--robustness frontier for the entire weighted class.
    Beyond the weighted class, for arbitrary deterministic strategyproof mechanisms with no computational restriction, we show the lower bound
    \(
        R\ge\max\left\{n,\frac{C}{C-1}\right\}
    \)
    for every \(C>1\), while \(C=1\) forces unbounded robustness.

    \item \textbf{Empirical evaluation (Section~\ref{sec:experiments}).}
    We evaluate the polynomial-time mechanisms on the Heterogeneous Computing Scheduling Problem (HCSP) benchmark instances of Braun et al.~\cite{braun2001comparison}, using predictions obtained by synthetically perturbing the true processing-time matrices.
    Across the 21 synthetic-noise settings, \textsc{EdgeSkip} improves on \textsc{ScaledGreedy} in 17, while \textsc{TolerantEdgeSkip} improves on \textsc{ErrorTolerantScaledGreedy} in 16; both comparisons favor our mechanisms at every sparse-error level.
\end{itemize}

\subsection{Additional Related Work}

In the literature on strategyproof scheduling, Nisan and Ronen proved the first lower bound of two for deterministic strategyproof mechanisms \cite{nisan1999algorithmic}.
It was subsequently improved through a sequence of works \cite{ashlagi2012optimal,christodoulou2009lower,koutsoupias2013lower,giannakopoulos2021new,dobzinski2020improved}, followed by a lower bound growing with the number of machines \cite{christodoulou2021nisan} and the eventual proof of the tight \(n\) lower bound \cite{christodoulou2023proof}.
Beyond the unrestricted unrelated-machines domain, Christodoulou, Koutsoupias, and Kov{\'a}cs studied strategyproof scheduling on graph- and hypergraph-restricted domains \cite{christodoulou2025truthful}.

A separate line of work studies algorithms and mechanisms with predictions.
The consistency--robustness viewpoint was developed to combine accurate predictions with worst-case guarantees. Representative early results include learning-augmented caching and general online problems \cite{purohit2018improving,lykouris2021competitive}, and an overview is given by \cite{mitzenmacher2022algorithms}.
Scheduling has been a central application, including work on predicted job sizes, learned weights, online unrelated-machine load balancing, and untrusted predictions \cite{mitzenmacher2020scheduling,lattanzi2020online,li2021online,bampis2022scheduling}.
These works study nonstrategic scheduling or online variants.

Predictions have also been used in strategic settings, including facility location and more general mechanism-design problems \cite{agrawal2022learning,xu2022mechanism}.
Most directly related to our work, Balkanski, Gkatzelis, and Tan \cite{balkanski2023strategyproof} studied strategyproof scheduling with predictions and developed the weighted mechanisms and error-tolerant construction that form the starting point for our upper bounds.
Cole, Gupta, and Jangir~\cite{cole2025parsimonious} replace the full \(n\times m\) processing-time matrix with only \(1+n+m\) predicted quantities: the optimal makespan, one dual-derived value for each machine, and one recommended machine for each job. Their deterministic strategyproof mechanism is \(C\)-consistent and \(n\bigl(2+1/(C-2)\bigr)\)-robust for every \(C>2\); in particular, it achieves the point \((C,R)=(4,2.5n)\).

\section{Preliminaries}\label{sec:preliminaries}

There is a set \(N=[n]\) of machines and a set \(M=[m]\) of jobs.
The processing time of job \(j\) on machine \(i\) is \(p_{ij}>0\), and we write \(\mathbf p=(p_{ij})_{i\in N,j\in M}\) for the resulting instance.
An allocation \(\mathbf x\in\{0,1\}^{n\times m}\) assigns every job to exactly one machine, subject to \(\sum_{i\in N}x_{ij}=1\) for each \(j\in M\).
We use the terms allocation and schedule interchangeably.
The load of machine \(i\) and the makespan of the allocation are
\begin{align*}
    L_i(\mathbf p,\mathbf x)=\sum_{j\in M}p_{ij}x_{ij},
    \qquad
    \MS(\mathbf p,\mathbf x)=\max_{i\in N}L_i(\mathbf p,\mathbf x).
\end{align*}
We use \(\mathbf x^*(\mathbf p)\) for a fixed optimal allocation that minimizes the makespan and write
\begin{align*}
    \OPT(\mathbf p)
    =\min_{\mathbf x}\MS(\mathbf p,\mathbf x)
    =\MS(\mathbf p,\mathbf x^*(\mathbf p)).
\end{align*}

In the strategic setting, each machine \(i\) is controlled by an agent whose private type is the vector \(\mathbf p_i=(p_{ij})_{j\in M}\).
A deterministic direct-revelation mechanism consists of an allocation rule \(\mathbf x\) and payments \(\boldsymbol\pi=(\pi_i)_{i\in N}\).
If the true type of machine \(i\) is \(\mathbf p_i\), its utility is its payment minus the true processing cost of the jobs assigned to it.
The mechanism is \emph{strategyproof}, or dominant-strategy incentive compatible (DSIC), if truthful reporting maximizes this utility for every machine and every profile of reports by the other machines.

We will often use \emph{weak monotonicity} in the arguments.
Fix the reports of all machines other than \(i\), and consider two reports \(\mathbf p_i\) and \(\mathbf p'_i\) by machine \(i\).
If \(\mathbf x_i\) and \(\mathbf x'_i\) are the corresponding allocation vectors of that machine, then weak monotonicity requires
\begin{align*}
    \sum_{j\in M}(p_{ij}-p'_{ij})(x_{ij}-x'_{ij})\le 0.
\end{align*}
It is known that a deterministic allocation rule admits DSIC payments if and only if it satisfies this weak monotonicity condition \cite{saks2005weak}.
Thus, as is standard in strategyproof scheduling \cite{nisan1999algorithmic}, we usually describe and analyze the allocation rule, while the corresponding payments are implicit.

\paragraph{Predictions, consistency, and robustness.} In the learning-augmented framework,
before the machines report their types, the mechanism receives a public prediction \(\widehat{\mathbf p}=(\widehat p_{ij})_{i\in N,j\in M}\) of the processing-time matrix.
The prediction is not controlled by the agents.
We write \(\mathbf x(\mathbf p;\widehat{\mathbf p})\) for the allocation on reports \(\mathbf p\) under prediction \(\widehat{\mathbf p}\).
For every fixed prediction \(\widehat{\mathbf p}\), the map \(\mathbf p\mapsto\mathbf x(\mathbf p;\widehat{\mathbf p})\) is required to be a DSIC allocation rule; equivalently, it must satisfy weak monotonicity.

The approximation ratio on the pair \((\mathbf p,\widehat{\mathbf p})\) is
\begin{align*}
    \rho(\mathbf p,\widehat{\mathbf p})=
    \frac{\MS(\mathbf p,\mathbf x(\mathbf p;\widehat{\mathbf p}))}{\OPT(\mathbf p)}.
\end{align*}
A mechanism is \(C\)-\emph{consistent} if its approximation ratio is at most \(C\) when the prediction is correct (i.e., \(\mathbf p=\widehat{\mathbf p}\)), and it is \(R\)-\emph{robust} if the same guarantee \(R\) holds for arbitrary predictions \cite{balkanski2023strategyproof}.
Formally,
\begin{align*}
    \sup_{\mathbf p}\rho(\mathbf p,\mathbf p)\le C,
    \qquad
    \sup_{\mathbf p,\widehat{\mathbf p}}\rho(\mathbf p,\widehat{\mathbf p})\le R.
\end{align*}
The case \(C=1\) is called perfect consistency.
For error-dependent guarantees, we use the standard symmetric multiplicative error:
\begin{align*}
    \eta(\mathbf p,\widehat{\mathbf p})
    =\max_{i\in N,j\in M}
    \left\{
        \frac{p_{ij}}{\widehat p_{ij}},
        \frac{\widehat p_{ij}}{p_{ij}}
    \right\}.
\end{align*}
Thus \(\eta\ge1\), with \(\eta=1\) exactly when the prediction is correct.
An error-tolerant mechanism is also given a tolerance parameter \(\bar\eta\ge1\);
its error-dependent guarantee applies when \(\eta(\mathbf p,\widehat{\mathbf p})\le\bar\eta\).

\paragraph{Reference schedules.}
Our mechanisms first compute a schedule from the prediction $\widehat{\mathbf p}$.
Let \(\widehat\sigma:M\to N\) be such a reference schedule.
We call $\widehat\sigma$ a \(\kappa\)-approximate reference schedule if its predicted makespan satisfies \(\max_{i\in N}\sum_{j:\widehat\sigma(j)=i}\widehat p_{ij}\le\kappa\OPT(\widehat{\mathbf p})\).
For each job \(j\), we call \(\widehat\sigma(j)\) its reference machine.
The case \(\kappa=1\) is useful for stating the tradeoff without a polynomial-time requirement.
For polynomial-time mechanisms, we use a schedule returned by the classical \(2\)-approximation algorithm, so \(\kappa=2\) \cite{lenstra1990approximation}.

\paragraph{Job-wise weighted mechanisms.}
Following \cite{balkanski2023strategyproof}, we use the term \emph{job-wise weighted mechanism} for the following class.
Let \(\mathbf r=(r_{ij})_{i\in N,j\in M}\) be a positive weight matrix that may depend on the prediction and on public mechanism parameters, but not on the machines' reports.
The associated weighted mechanism \(M_{\mathbf r}\) assigns each job independently according to
\begin{align*}
    M_{\mathbf r}(\mathbf p)(j)
    \in\argmin_{i\in N}\{r_{ij}p_{ij}\}.
\end{align*}
We call \(r_{ij}p_{ij}\) the \emph{weighted score} of machine \(i\) for job \(j\).
All ties are resolved by a deterministic rule fixed before the machines report their types; the mechanisms developed below give priority to a machine selected from the prediction.
For a fixed prediction and job, and fixed reports by the other machines, there is a threshold below which machine \(i\) wins the job; equality is resolved by the fixed tie-breaking rule.
Paying this threshold for each allocated job and summing the payments over jobs gives a DSIC implementation.

\section{Improved Upper Bounds}\label{sec:upper-bound}

This section is organized as follows. We first introduce \textsc{EdgeSkip} and state its consistency and robustness guarantees. We then prove these guarantees, using a partial-matching reduction and a forest-charging argument for robustness. Finally, we develop an error-tolerant variant and establish both its error-dependent approximation guarantee and its unconditional robustness guarantee.

\subsection{The EdgeSkip Mechanism}\label{sec:edge-skip}

We introduce a new job-wise weighted mechanism, called \textsc{EdgeSkip}.
Fix \(\gamma>0\) and a \(\kappa\)-approximate reference schedule \(\widehat\sigma\), and let \(T=\max_{i\in N}\sum_{j:\widehat\sigma(j)=i}\widehat p_{ij}\) be its predicted makespan.
Our mechanism first constructs a target machine \(y(j)\) for every job.
Initially, the target is the reference machine \(\widehat\sigma(j)\).
For each machine \(i\), a counter \(I_i\) starts at zero and records the predicted load of the jobs moved into \(i\).
The mechanism considers every pair \((i,j)\) with \(\widehat p_{ij}<\widehat p_{\widehat\sigma(j),j}\), in nonincreasing order of the predicted speedup \(\frac{\widehat p_{\widehat\sigma(j),j}}{\widehat p_{ij}}\).
It moves \(j\) to \(i\) only if \(j\) has not moved before and the new incoming load would remain at most \(\gamma T\).
We call the inequality \(I_i+\widehat p_{ij}\le\gamma T\) the \emph{capacity test} for pair \((i,j)\).

\begin{algorithm}[H]
\caption{\textsc{EdgeSkip}}\label{alg:edge-skip}
\KwIn{Prediction \(\widehat{\mathbf p}\), reports \(\mathbf p\), reference schedule \(\widehat\sigma\), and parameter \(\gamma>0\)}
Set \(T\leftarrow\max_{i\in N}\sum_{j:\widehat\sigma(j)=i}\widehat p_{ij}\) to be
the predicted makespan of \(\widehat\sigma\)\;
Set \(y(j)\leftarrow\widehat\sigma(j)\) for every \(j\in M\), and \(I_i\leftarrow0\) for every \(i\in N\)\;
Let \(E=\{(i,j):\widehat p_{ij}<\widehat p_{\widehat\sigma(j),j}\}\)\;
Order \(E\) by nonincreasing \(\frac{\widehat p_{\widehat\sigma(j),j}}{\widehat p_{ij}}\), with ties fixed in advance\;
\For{\((i,j)\in E\), in this order}{
    \If{\(y(j)=\widehat\sigma(j)\) and \(I_i+\widehat p_{ij}\le\gamma T\)}{
        \(y(j)\leftarrow i\) and \(I_i\leftarrow I_i+\widehat p_{ij}\)\;
    }
}
Set \(r^0_{ij}\leftarrow\max\{1,\frac{\widehat p_{y(j),j}}{\widehat p_{ij}}\}\) for every \(i,j\)\;
Assign \(j\) to a machine minimizing weighted score \(r^0_{ij}p_{ij}\), giving \(y(j)\) priority in a tie\;
\end{algorithm}

The resulting assignment \(y\) is the \emph{target schedule}, and \(\mathbf r^0\) is the corresponding \emph{base weight matrix}.
For each job \(j\), a machine with \(\widehat p_{ij}\ge\widehat p_{y(j),j}\) receives weight one, while a machine with \(\widehat p_{ij}<\widehat p_{y(j),j}\) receives weight \(\frac{\widehat p_{y(j),j}}{\widehat p_{ij}}\).
Thus
\(
    r^0_{ij}\widehat p_{ij}=\max\{\widehat p_{ij},\widehat p_{y(j),j}\}\ge\widehat p_{y(j),j}.
\)
For \(i=y(j)\), equality holds, and the tie-breaking rule assigns \(j\) to \(y(j)\) when \(\mathbf p=\widehat{\mathbf p}\); under other reports, a different machine may minimize \(r^0_{ij}p_{ij}\).

\textsc{EdgeSkip} and \textsc{ScaledGreedy} \cite{balkanski2023strategyproof} examine the same machine--job pairs in the same order and differ only when a candidate pair \((i,j)\) fails the capacity test, that is, when \(I_i+\widehat p_{ij}>\gamma T\). \textsc{ScaledGreedy} accepts this pair and then closes machine \(i\), so its incoming load may exceed \(\gamma T\). In contrast, \textsc{EdgeSkip} rejects only \((i,j)\) and continues scanning the remaining pairs. The job \(j\) stays on its reference machine unless a later pair moves it elsewhere, and machine \(i\) remains available for later jobs. As a result, \textsc{EdgeSkip} maintains \(I_i\le\gamma T\) throughout the construction, which yields the improved consistency bound.

However, this change makes the robustness analysis less direct. A failed capacity test implies only \(I_i+\widehat p_{ij}>\gamma T\),
not \(I_i\ge\gamma T\): the rejected job itself may account for most of the excess. Thus the accepted incoming load on \(i\) cannot be
 charged by itself. The robustness proof below instead accounts for the accepted jobs together with the rejected job and uses a dependency forest to control repetitions across different rejected pairs.

\begin{theorem}\label{thm:edge-skip}
For every \(\kappa\)-approximate reference schedule and every \(C>\kappa\), run \textnormal{\textsc{EdgeSkip}} with
\(\gamma=(C-\kappa)/\kappa\).
The resulting mechanism is strategyproof, \(C\)-consistent, and \(R_{\kappa}(C)\)-robust, where
\begin{align}\label{eq:robustness-function}
    R_{\kappa}(C)
    =
    \begin{cases}
        \displaystyle \max\left\{n,n-1+\frac{\kappa n-1}{C-\kappa}\right\},
        &\kappa<C\le\kappa+1,\\[2mm]
        \displaystyle \max\left\{n,n-2+\frac{\kappa n}{C-\kappa}\right\},
        &\kappa+1<C\le2\kappa,\\[2mm]
        \displaystyle \max\left\{n,n-1+\frac{\kappa(n-1)}{C-\kappa}\right\},
        &C>2\kappa.
    \end{cases}
\end{align}
\end{theorem}

The two relevant reference schedules are an optimal schedule (\(\kappa=1\)) and the best polynomial-time schedule used in our results, namely the standard \(2\)-approximate schedule (\(\kappa=2\)).

\begin{corollary}\label{cor:edge-skip-specializations}
With \(\kappa=1\) and \(\gamma=C-1\), \textnormal{\textsc{EdgeSkip}} is \(C\)-consistent and
\(\max\{n,\frac{(n-1)C}{C-1}\}\)-robust for every \(C>1\).
With \(\kappa=2\) and \(\gamma=1\), it is polynomial-time, \(4\)-consistent, and \((2n-2)\)-robust.
\end{corollary}

The \(\kappa=1\) guarantee matches the information-theoretic lower bound in Section~\ref{sec:weighted-lower-bound} for every \(C>1\), and therefore gives the exact consistency--robustness frontier for job-wise weighted mechanisms.

The \(\kappa=2\) guarantee is our main polynomial-time point \((C,R)=(4,2n-2)\), which improves the \((6,2n)\) guarantee of \textsc{ScaledGreedy}~\cite{balkanski2023strategyproof}. At the same consistency level, it also improves the \((4,2.5n)\) consistency--robustness guarantee of Cole, Gupta, and Jangir (CGJ)~\cite{cole2025parsimonious}. This comparison concerns guarantees under different prediction interfaces: the CGJ mechanism is polynomial-time once its predicted value of \(\OPT\), machine-specific dual values, and one predicted machine per job are supplied, whereas \textsc{EdgeSkip} computes its \(\kappa=2\) reference schedule in polynomial time from the full predicted processing-time matrix.

\subsection{Analysis of EdgeSkip}
We first establish the two simpler parts of Theorem~\ref{thm:edge-skip}.
Their proofs are short because all weights are fixed before the machines submit their reports and because the capacity test directly bounds the target schedule.

\begin{lemma}\label{lem:edge-skip-basic}
\textnormal{\textsc{EdgeSkip}} is strategyproof and \((1+\gamma)\kappa\)-consistent.
\end{lemma}

\begin{proof}
The weights \(\mathbf r^0\) depend only on the prediction, the reference schedule, and \(\gamma\), not on the machines' reports.
Fix the reports of all machines other than \(i\), and consider two reports \(\mathbf p_i\) and \(\mathbf p'_i\) by machine \(i\), with corresponding allocation vectors \(\mathbf x_i\) and \(\mathbf x'_i\).
For every job \(j\), if the two reports differ on \(p_{ij}\), the lower reported processing time can only cause machine \(i\) to gain the job, while the higher one can only cause it to lose the job.
Hence
\(
    (p_{ij}-p'_{ij})(x_{ij}-x'_{ij})\le 0.
\)
Summing over all jobs gives weak monotonicity, and therefore the allocation rule admits strategyproof payments.

Suppose now that \(\mathbf p=\widehat{\mathbf p}\).
The target machine \(y(j)\) has weight \(r^0_{y(j),j}=1\), so its weighted score is \(\widehat p_{y(j),j}\).
For every machine \(i\), we have \(r^0_{ij}\widehat p_{ij}=\max\{\widehat p_{ij},\widehat p_{y(j),j}\}\ge\widehat p_{y(j),j}\).
The target-priority tie-breaking rule therefore assigns every job \(j\) to \(y(j)\).

Consider the load of a machine \(i\) in this target schedule.
The jobs that remain from its reference schedule have total processing time at most \(T\).
Every other job assigned to \(i\) was counted in \(I_i\), and the capacity test maintains \(I_i\le\gamma T\).
The target load of \(i\) is consequently at most \((1+\gamma)T\).
Since \(\mathbf p=\widehat{\mathbf p}\) and \(T\le\kappa\OPT(\widehat{\mathbf p})\), this load is at most \((1+\gamma)\kappa\OPT(\mathbf p)\), proving the claimed consistency.
\end{proof}

The consistency of \textsc{ScaledGreedy} \cite{balkanski2023strategyproof} is \((2+\gamma)\kappa\).
Before \textsc{ScaledGreedy} accepts a pair \((i,j)\) that fails \textsc{EdgeSkip}'s capacity test, machine \(i\)'s incoming load is at most \(\gamma T\), and job \(j\) has predicted processing time less than \(\widehat p_{\widehat\sigma(j),j}\le T\).
Thus its incoming load can reach \((1+\gamma)T\), and adding the remaining reference load gives a target load of at most \((2+\gamma)T\).
In contrast, \textsc{EdgeSkip} keeps the incoming load at most \(\gamma T\), so its target load is at most \((1+\gamma)T\).
Thus the improvement in consistency is not caused by a different weighted allocation rule.
It comes entirely from skipping a pair that would exceed the incoming-load limit while constructing the target schedule.
For the standard \(2\)-approximate reference schedule and \(\gamma=1\), this changes the consistency factor from \(6\) to \(4\).

Next, we prove the robustness guarantee in Theorem~\ref{thm:edge-skip} through a partial-matching bound on the load of each machine in the output allocation.
Fix arbitrary reports \(\mathbf p\), let \(\mathbf x\) be the resulting allocation, and let \(\mathbf x^*\) be an optimal allocation for \(\mathbf p\).
Fix an arbitrary machine \(h\in N\), and set
\[
    J_i=\{j:x_{hj}=1,\ x^*_{ij}=1\}.
\]
Thus, \(J_i\) consists exactly of the jobs assigned to \(h\) by the mechanism and to \(i\) by the optimal allocation.
For every nonempty \(J_i\), choose \(j_i\in J_i\) maximizing \(r^0_{ij}/r^0_{hj}\).
The weighted allocation rule gives
\[
\begin{aligned}
    \sum_{j\in J_i}p_{hj}
    &\le
    \left(\max_{j\in J_i}\frac{r^0_{ij}}{r^0_{hj}}\right)
    \left(\sum_{j\in J_i}p_{ij}\right)\\
    &\le
    \frac{r^0_{ij_i}}{r^0_{hj_i}}\OPT(\mathbf p).
\end{aligned}
\]
The pairs \((i,j_i)\) have distinct machines and distinct jobs, so they form a partial matching.
Consequently, it is enough to bound the weight
\[
    w(Q)=\sum_{(i,j)\in Q}\frac{r^0_{ij}}{r^0_{hj}}
\]
of every partial matching \(Q\) between machines and jobs, for every choice of \(h\).

\begin{lemma}[Excess-weight forest bound]\label{lem:edge-skip-forest-bound}
Fix a machine \(h\) and a partial matching \(Q\). For each edge \(e=(i,j)\in Q\), define its weight by
\(
    w_e:=\frac{r^0_{ij}}{r^0_{hj}}.
\)
Suppose \(w_e>1\) for some \(e\in Q\).
There exist a nonempty set \(\mathcal L\subseteq\{e\in Q:w_e>1\}\) and values \(P_e\ge1\), \(e\in\mathcal L\), such that
\begin{align}
    \sum_{\substack{e\in Q:w_e>1}}(w_e-1)
    &\le \sum_{e\in\mathcal L}(P_e-1),
    \label{eq:edge-skip-excess}\\
    \sum_{e\in\mathcal L}P_e
    &<\frac n\gamma,
    \label{eq:edge-skip-basic-budget}\\
    \sum_{e\in\mathcal L}P_e
    &\le\frac n\gamma-
    \max\left\{0,\frac1{\kappa\gamma}-1\right\},
    \label{eq:edge-skip-refined-budget}\\
    \mathcal L=\{e\}
    \quad &\Longrightarrow\quad
    P_e<\frac{n-1}{\gamma}.
    \label{eq:edge-skip-one-leaf}
\end{align}
\end{lemma}

\begin{proof}
Consider \(e=(i,j)\in Q\) with \(w_e>1\).
Since \(r^0_{hj}\ge1\), we have \(r^0_{ij}>1\), and hence
\[
    r^0_{ij}
    =\frac{\widehat p_{y(j),j}}{\widehat p_{ij}},
    \qquad
    \widehat p_{ij}
    <\widehat p_{y(j),j}
    \le\widehat p_{\widehat\sigma(j),j}.
\]
Hence, \((i,j)\in E\) and was therefore considered by \textsc{EdgeSkip}.
If the target-schedule construction moved \(j\) from \(\widehat\sigma(j)\) to its final target \(y(j)\), then \((y(j),j)\) is the pair that performed this move.
Because \(\widehat p_{ij}<\widehat p_{y(j),j}\), the pair \((i,j)\) has strictly larger predicted speedup than \((y(j),j)\) and was therefore processed first, while \(j\) was still available.
If the construction never moved \(j\), then \(j\) likewise remained available when \((i,j)\) was processed.
In either case, \((i,j)\) was processed while \(j\) was available but did not move \(j\) to \(i\); it must therefore have been rejected by the capacity test.

Let \(A_e\) be the set of jobs \(\ell\) for which the pair \((i,\ell)\) was accepted before \((i,j)\) was processed.
Equivalently, these are precisely the jobs that had already been moved from their reference machines into \(i\) when \((i,j)\) failed the capacity test.
Thus, immediately before this test,
\(
    I_i=\sum_{\ell\in A_e}\widehat p_{i\ell},
\)
and \(j\notin A_e\).
Put
\[
    a_e=\frac{\widehat p_{\widehat\sigma(j),j}}{\widehat p_{ij}}.
\]
The rejection and the global processing order give
\begin{align}
    \widehat p_{ij}+\sum_{\ell\in A_e}\widehat p_{i\ell}
    &>\gamma T,
    \label{eq:edge-skip-rejection}\\
    \frac{\widehat p_{\widehat\sigma(\ell),\ell}}
         {\widehat p_{i\ell}}
    &\ge a_e
    \qquad(\ell\in A_e),
    \label{eq:edge-skip-order}
\end{align}
while
\[
    w_e
    \le
    \frac{\widehat p_{y(j),j}}{\widehat p_{ij}}
    \le a_e.
\]

Make one vertex for every edge \(e\in Q\) with \(w_e>1\), and draw an arc
\(e\to e'\) when \(j\in A_{e'}\).
Each vertex has at most one outgoing arc because a job moves at most once and the machine endpoints of the edges in \(Q\) are distinct.
Every arc follows the processing order, so there is no directed cycle.
Moreover, an undirected cycle together with the outdegree-one condition would force a directed cycle.
The resulting graph is therefore a forest oriented toward its roots.
Let \(\mathcal L\) be its set of leaves, where a leaf has indegree zero, and let \(P_e\) be the product of the vertex weights on the path from \(e\) to its root.

A backward induction gives \(P_e\le a_e\) for every vertex.
The claim holds at a root because \(w_e\le a_e\).
If \(e=(i,j)\) has parent \(e'\), then \(j\in A_{e'}\), so \(j\) was moved to \(y(j)\) before \(e'\) was rejected.
The processing order gives
\[
    \frac{\widehat p_{\widehat\sigma(j),j}}
         {\widehat p_{y(j),j}}
    \ge a_{e'}.
\]
Assuming \(P_{e'}\le a_{e'}\), we obtain
\[
\begin{aligned}
    P_e
    &=w_e P_{e'}\\
    &\le
    \frac{\widehat p_{y(j),j}}{\widehat p_{ij}}a_{e'}\\
    &\le
    \frac{\widehat p_{y(j),j}}{\widehat p_{ij}}
    \frac{\widehat p_{\widehat\sigma(j),j}}
         {\widehat p_{y(j),j}}
    =a_e.
\end{aligned}
\]

For a leaf \(e=(i_e,j_e)\), define
\[
    C_e=\{j_e\}\cup A_e,\qquad
    x_e=\sum_{j\in C_e}\widehat p_{i_e j},\qquad
    W_e=\sum_{j\in C_e}\widehat p_{\widehat\sigma(j),j}.
\]
Equations~\eqref{eq:edge-skip-rejection}--\eqref{eq:edge-skip-order} and \(P_e\le a_e\) imply
\begin{align}
    x_e>\gamma T,
    \qquad
    W_e\ge P_e x_e.
    \label{eq:edge-skip-leaf-charge}
\end{align}
The sets \(C_e\) for \(e\in\mathcal L\) are pairwise disjoint.
Indeed, the sets \(A_e\) are disjoint for distinct machines \(i_e\), the jobs \(j_e\) are distinct because \(Q\) is a matching, and an inclusion \(j_e\in A_{e'}\) would give the leaf \(e'\) an incoming arc.
Consequently,
\[
    \gamma T\sum_{e\in\mathcal L}P_e
    <\sum_{e\in\mathcal L}W_e
    \le\sum_{j\in M}\widehat p_{\widehat\sigma(j),j}
    \le nT,
\]
which proves~\eqref{eq:edge-skip-basic-budget}.
If there is only one leaf \(e=(i_e,j_e)\), no job in \(C_e\) has reference machine \(i_e\), since every such job corresponds to a strictly faster pair into \(i_e\).
Hence \(W_e\le(n-1)T\), which proves~\eqref{eq:edge-skip-one-leaf}.

For the refined budget, remove all jobs in the sets \(C_e\), \(e\in\mathcal L\), from the reference schedule, and let \(b_i\) be the remaining load of machine \(i\).
Reassign every \(C_e\) to machine \(i_e\), leaving all other jobs on their reference machines.
This is feasible because the sets \(C_e\) are pairwise disjoint and the machines \(i_e\) are distinct.
Its load is \(b_{i_e}+x_e\) on each machine \(i_e\) and \(b_i\) on every other machine.
Since \(T\le\kappa\OPT(\widehat{\mathbf p})\), every feasible schedule has makespan at least \(T/\kappa\). Thus, the reassigned schedule has some machine load of at least \(T/\kappa\).
Writing \(z_e=x_e-\gamma T>0\) and using \(P_e\ge1\), we get
\[
    \sum_i b_i+\sum_{e\in\mathcal L}P_e z_e
    \ge
    \max\left\{0,\left(\frac1\kappa-\gamma\right)T\right\}.
\]
By the definitions of \(W_e\) and \(b_i\), the total predicted processing time in the reference schedule decomposes as
\[
\begin{aligned}
    nT
    &\ge\sum_{j\in M}\widehat p_{\widehat\sigma(j),j}\\
    &=\sum_{e\in\mathcal L}W_e+\sum_i b_i\\
    &\ge
    \gamma T\sum_{e\in\mathcal L}P_e
    +\sum_{e\in\mathcal L}P_e z_e+\sum_i b_i\\
    &\ge
    \gamma T\sum_{e\in\mathcal L}P_e
    +\max\left\{0,\left(\frac1\kappa-\gamma\right)T\right\},
\end{aligned}
\]
which proves~\eqref{eq:edge-skip-refined-budget}.

It remains to prove the excess inequality.
For any rooted tree with vertex weights at least one,
\[
    \sum_v(w_v-1)\le\sum_{e\text{ leaf}}(P_e-1).
\]
For a one-vertex tree this is equality.
Otherwise remove a root of weight \(w\), and let \(P'_e\) be the leaf products in the remaining subtrees.
After applying induction to those subtrees, the required step is
\[
    (w-1)+\sum_e(P'_e-1)
    \le\sum_e(wP'_e-1),
\]
which is equivalent to
\((w-1)(\sum_e P'_e-1)\ge0\).
Summing over the tree components proves~\eqref{eq:edge-skip-excess}.
\end{proof}

\begin{proof}[Proof of the robustness guarantee in Theorem~\ref{thm:edge-skip}]
Fix a machine \(h\) and consider an arbitrary partial matching \(Q\).
If every edge of \(Q\) has weight at most one, then \(w(Q)\le |Q|\le n\).
Otherwise, \(\mathcal L\) is nonempty.
Equation~\eqref{eq:edge-skip-excess} gives
\begin{align}
    w(Q)
    &=\sum_{e\in Q}\min\{w_e,1\}
      +\sum_{\substack{e\in Q:w_e>1}}(w_e-1)\notag\\
    &\le |Q|+\sum_{e\in\mathcal L}(P_e-1)\notag\\
    &\le n-|\mathcal L|+\sum_{e\in\mathcal L}P_e.
    \label{eq:edge-skip-master}
\end{align}
If \(|\mathcal L|=1\), equation~\eqref{eq:edge-skip-one-leaf} gives
\[
    w(Q)<n-1+\frac{n-1}{\gamma}.
\]
If \(|\mathcal L|\ge2\), equations~\eqref{eq:edge-skip-master} and~\eqref{eq:edge-skip-refined-budget} give
\[
    w(Q)
    \le n-2+\frac n\gamma
       -\max\left\{0,\frac1{\kappa\gamma}-1\right\}.
\]
For \(0<\gamma\le1/\kappa\), the bound obtained when \(|\mathcal L|\ge2\) is
\(n-1+(n-1/\kappa)/\gamma\), and it is at least the bound obtained when \(|\mathcal L|=1\).
For \(1/\kappa<\gamma\le1\), the bound obtained when \(|\mathcal L|\ge2\) is \(n-2+n/\gamma\), and it is again at least the bound obtained when \(|\mathcal L|=1\).
For \(\gamma>1\), the bound obtained when \(|\mathcal L|=1\) is the larger one, because it exceeds the bound obtained when \(|\mathcal L|\ge2\) by \(1-1/\gamma\).
Combining these three cases with the initial bound \(w(Q)\le n\) when all edge weights are at most one gives
\[
    w(Q)\le R_{\kappa}((1+\gamma)\kappa)
\]
for every partial matching \(Q\).

Let
\(
    Q_h=\{(i,j_i):J_i\ne\varnothing\}
\)
be the partial matching constructed in the paragraph preceding
Lemma~\ref{lem:edge-skip-forest-bound}.
Recalling that \(L_h(\mathbf p,\mathbf x)=\sum_{j:x_{hj}=1}p_{hj}\) is the load of machine \(h\) under allocation \(\mathbf x\), summing the preceding bounds over \(i\) gives
\[
\begin{aligned}
    L_h(\mathbf p,\mathbf x)
    &=\sum_{i:J_i\ne\varnothing}\sum_{j\in J_i}p_{hj}\\
    &\le
    \sum_{i:J_i\ne\varnothing}
    \frac{r^0_{ij_i}}{r^0_{hj_i}}
    \sum_{j\in J_i}p_{ij}\\
    &\le
    \sum_{(i,j_i)\in Q_h}
    \frac{r^0_{ij_i}}{r^0_{hj_i}}\OPT(\mathbf p)\\
    &=w(Q_h)\OPT(\mathbf p)\\
    &\le R_{\kappa}((1+\gamma)\kappa)\OPT(\mathbf p).
\end{aligned}
\]
Taking the maximum over \(h\) yields the robustness claim.
\end{proof}

\subsection{Error-Dependent Guarantees}\label{sec:error-tolerant}

Fix a tolerance parameter \(\bar\eta\ge1\), chosen in advance.
It specifies the largest multiplicative prediction error for which the mechanism is designed to retain an error-dependent guarantee: this guarantee applies when
\(\eta(\mathbf p,\widehat{\mathbf p})\le\bar\eta\), whereas the unconditional robustness guarantee must hold for arbitrary predictions.
Let \(\mathbf r^0\) and \(y\) be the weights and target schedule produced in \textsc{EdgeSkip}, and define
\begin{align*}
    \widetilde r_{ij}=
    \begin{cases}
        \frac{1}{\bar\eta^2},&i=y(j),\\
        r^0_{ij},&i\ne y(j).
    \end{cases}
\end{align*}
We call the resulting weighted mechanism \textsc{TolerantEdgeSkip}.
It assigns each job \(j\) to a machine minimizing \(\widetilde r_{ij}p_{ij}\), giving \(y(j)\) priority in a tie.
The same modification is used in \textsc{ErrorTolerantScaledGreedy} \cite{balkanski2023strategyproof}; here it is applied to the target schedule of \textsc{EdgeSkip}.

\begin{theorem}\label{thm:error-tolerant}
\textnormal{\textsc{TolerantEdgeSkip}} is strategyproof.
If \(\eta(\mathbf p,\widehat{\mathbf p})\le\bar\eta\), its approximation ratio is at most \((1+\gamma)\kappa\eta^2\).
For arbitrary predictions, its approximation ratio is at most
\(
    1+\bar\eta^2\bigl(R_{\kappa}((1+\gamma)\kappa)-1\bigr).
\)
\end{theorem}

\begin{proof}
The weights \(\widetilde{\mathbf r}\) depend only on public information, so \textsc{TolerantEdgeSkip} is strategyproof.
Let \(\eta=\eta(\mathbf p,\widehat{\mathbf p})\le\bar\eta\).
Fix a job \(j\).
For the target machine \(y(j)\),
\begin{align*}
    \widetilde r_{y(j),j}p_{y(j),j}
    \le\frac{\eta\widehat p_{y(j),j}}{\bar\eta^2}
    \le\frac{\widehat p_{y(j),j}}{\bar\eta}.
\end{align*}
For any machine \(i\ne y(j)\),
\begin{align*}
    \widetilde r_{ij}p_{ij}
    =r^0_{ij}p_{ij}
    \ge\frac{r^0_{ij}\widehat p_{ij}}{\eta}
    =\frac{\max\{\widehat p_{y(j),j},\widehat p_{ij}\}}{\eta}
    \ge\frac{\widehat p_{y(j),j}}{\bar\eta}.
\end{align*}
Therefore, \(
    \widetilde r_{y(j),j}p_{y(j),j}
    \le \widetilde r_{ij}p_{ij},
\) and the tie-breaking rule selects \(y(j)\).
Since this holds for every job, \textsc{TolerantEdgeSkip} returns the schedule \(y\).
As shown in the proof of Lemma~\ref{lem:edge-skip-basic}, the load of every machine in \(y\), evaluated at \(\widehat{\mathbf p}\), is at most \((1+\gamma)T\).
Moreover, \(p_{ij}\le\eta\widehat p_{ij}\) for every \(i,j\).
Hence, for every machine \(h\), the actual load is
\[
    \sum_{j:y(j)=h}p_{hj}
    \le \eta\sum_{j:y(j)=h}\widehat p_{hj}
    \le \eta(1+\gamma)T.
\]
Taking the maximum over \(h\), the makespan of \(y\) under \(\mathbf p\) is at most \(\eta(1+\gamma)T\).
If \(\mathbf x^*(\mathbf p)\) is optimal for the actual instance, its predicted makespan is at most \(\eta\OPT(\mathbf p)\), and hence \(\OPT(\widehat{\mathbf p})\le\eta\OPT(\mathbf p)\).
Using \(T\le\kappa\OPT(\widehat{\mathbf p})\) proves the error-dependent guarantee \((1+\gamma)\kappa\eta^2\).

For the unconditional guarantee, let \(\mathbf x\) be the allocation produced on \((\mathbf p,\widehat{\mathbf p})\). For every \(i,h,j\), observe that
\begin{align*}
    \frac{\widetilde r_{ij}}{\widetilde r_{hj}}
    \le\bar\eta^2\frac{r^0_{ij}}{r^0_{hj}}.
\end{align*}
Fix a machine \(h\) and a partial matching \(Q\).
If \((h,j)\in Q\), that edge contributes \(1\), and applying the preceding ratio bound to the other edges and \(w(Q)\le R_{\kappa}((1+\gamma)\kappa)\) gives
\[
    \sum_{(i,j)\in Q}\frac{\widetilde r_{ij}}{\widetilde r_{hj}}
    \le1+\bar\eta^2\bigl(R_{\kappa}((1+\gamma)\kappa)-1\bigr).
\]
If \(Q\) contains no such edge, then \(|Q|\le n-1\); using this in~\eqref{eq:edge-skip-master} gives \(w(Q)\le R_{\kappa}((1+\gamma)\kappa)-1\), and hence
\[
    \sum_{(i,j)\in Q}\frac{\widetilde r_{ij}}{\widetilde r_{hj}}
    \le\bar\eta^2\bigl(R_{\kappa}((1+\gamma)\kappa)-1\bigr).
\]
Set \(J_i=\{j:x_{hj}=1,\ x^*_{ij}=1\}\).
For every nonempty \(J_i\), choose \(j_i\) maximizing \(\widetilde r_{ij}/\widetilde r_{hj}\); then \(Q_h=\{(i,j_i):J_i\ne\varnothing\}\) is a partial matching.
Since \(\widetilde r_{hj}p_{hj}\le\widetilde r_{ij}p_{ij}\) for every \(j\in J_i\),
\begin{align*}
    L_h(\mathbf p,\mathbf x)
    &=\sum_{i:J_i\ne\varnothing}\sum_{j\in J_i}p_{hj}\\
    &\le\sum_{i:J_i\ne\varnothing}
       \frac{\widetilde r_{ij_i}}{\widetilde r_{hj_i}}
       \sum_{j\in J_i}p_{ij}\\
    &\le\sum_{(i,j_i)\in Q_h}
       \frac{\widetilde r_{ij_i}}{\widetilde r_{hj_i}}\OPT(\mathbf p)\\
    &\le
    \left[1+\bar\eta^2\bigl(R_{\kappa}((1+\gamma)\kappa)-1\bigr)\right]
    \OPT(\mathbf p).
\end{align*}
Taking the maximum over \(h\) proves the unconditional robustness bound.
\end{proof}

For general \(\kappa\) and \(\gamma\), the corresponding guarantees of \textsc{ErrorTolerantScaledGreedy} are \((2+\gamma)\kappa\eta^2\) when \(\eta\le\bar\eta\) and \(n(1+\frac{1}{\gamma})\bar\eta^2\) for arbitrary predictions \cite{balkanski2023strategyproof}.
\textsc{TolerantEdgeSkip} strictly improves the first bound to \((1+\gamma)\kappa\eta^2\) and the second to
\(1+\bar\eta^2\bigl(R_{\kappa}((1+\gamma)\kappa)-1\bigr)\).

With a \(2\)-approximate reference schedule and \(\gamma=1\), \textsc{TolerantEdgeSkip} is a \(4\eta^2\)-approximation whenever \(\eta\le\bar\eta\), and an unconditional \(1+(2n-3)\bar\eta^2\)-approximation.
The corresponding bounds for \textsc{ErrorTolerantScaledGreedy} are \(6\eta^2\) and \(2n\bar\eta^2\) \cite{balkanski2023strategyproof}.

\section{Information-Theoretic Lower Bounds}\label{sec:lower-bounds}

This section presents two lower bounds that do not impose any computational restriction. We first use a simple instance with three jobs to prove a lower bound for all deterministic strategyproof mechanisms. We then use a replicated-job construction to prove that every \(C\)-consistent job-wise weighted mechanism has robustness at least \(\max\{n,(n-1)C/(C-1)\}\), exactly matching the upper bound achieved by \textsc{EdgeSkip}.

\subsection{A General Lower Bound}\label{sec:basic-lower-bound}

There are two independent obstacles to the robustness guarantee $R$.
First, fixing any prediction \(\widehat{\mathbf p}\) leaves an ordinary deterministic strategyproof scheduling rule \(\mathbf x_{\widehat{\mathbf p}}(\mathbf p):=\mathbf x(\mathbf p;\widehat{\mathbf p})\).
For \(n\) unrelated machines, no deterministic strategyproof scheduling mechanism can guarantee a makespan approximation ratio better than \(n\) \cite{christodoulou2023proof}.
Therefore, \(R\ge n\).
The next theorem gives the second obstacle using only two machines and three jobs.

\begin{theorem}\label{thm:basic-lower-bound}
Let a deterministic strategyproof mechanism be \(C\)-consistent and \(R\)-robust on \(n\ge2\) unrelated machines.
If \(C>1\), then
\(
    R\ge \max\left\{n,\frac{C}{C-1}\right\}.
\)
If \(C=1\), then its robustness is unbounded.
\end{theorem}

\begin{proof}
The bound \(R\ge n\) follows from the information-theoretic lower bound of \cite{christodoulou2023proof} for \(n\) unrelated machines without predictions.
It remains to prove the consistency-dependent term.
First suppose that \(C>1\).
Choose \(0<K<\frac{1}{C-1}\), and then choose \(0<\delta<\min\{K,1\}\) small enough that \(K+1>C(K+\delta)\). Also choose \(0<\xi<K\).
It is enough to use two machines and three jobs \(a,b,c\); on any additional machine, all three jobs have a sufficiently large processing time \(H\).
Consider the prediction
\begin{align*}
    \widehat{\mathbf p}=
    \begin{array}{c|ccc}
        &a&b&c\\ \hline
        1&K&1&H\\
        2&H&K&\delta
    \end{array}.
\end{align*}
Its optimal makespan is \(K+\delta\), attained by assigning \(a\) to machine 1 and \(b,c\) to machine 2.
For large enough \(H\), consistency forces \(a\) and \(c\) to these machines.
It also forces \(b\) to machine 2, since assigning it to machine 1 would give makespan \(K+1\), and \(K+1>C(K+\delta)\) violates the consistency requirement.

Keep this prediction fixed and first change only the report of machine 1 from \((K,1,H)\) to \((\delta,1+\xi,H)\).
Let \(x'_{1a}\) and \(x'_{1b}\) be its new allocation indicators.
Weak monotonicity gives
\begin{align*}
    (K-\delta)(1-x'_{1a})+\xi x'_{1b}\le0.
\end{align*}
Both terms are nonnegative, so machine 1 continues to receive \(a\) and not \(b\).
This is the first locking step: the processing time of the job already assigned to machine 1 is decreased, while that of the unassigned job is increased.
If \(b\) or \(c\) is now assigned to a machine on which its processing time is \(H\), the approximation ratio is already arbitrarily large as \(H\) grows.
We may therefore continue with the only remaining case, in which machine 2 receives both jobs.

Keeping the modified report \((\delta,1+\xi,H)\) of machine 1 fixed, change only the report of machine 2 from \((H,K,\delta)\) to \((H,K-\xi,1)\).
The job \(c\) remains on machine 2: otherwise the mechanism already has an arbitrarily large approximation ratio as \(H\) grows.
Writing \(x''_{2b}\) and \(x''_{2c}\) for the new allocation indicators, weak monotonicity for this second change gives
\begin{align*}
    \xi(1-x''_{2b})-(1-\delta)(1-x''_{2c})\le0.
\end{align*}
We have just argued that \(x''_{2c}=1\), and hence this inequality forces \(x''_{2b}=1\).
Thus, on the final instance, machine 2 receives both \(b\) and \(c\), and the mechanism has makespan at least \(K-\xi+1\).
The allocation that assigns \(a,b\) to machine 1 and \(c\) to machine 2 has makespan \(1+\xi+\delta\).
Consequently,
\begin{align*}
    R\ge \frac{K+1-\xi}{1+\xi+\delta}.
\end{align*}
Letting \(\xi\) and \(\delta\) tend to zero and \(K\) tend to \(\frac{1}{C-1}\) proves \(R\ge1+\frac{1}{C-1}=\frac{C}{C-1}\).
When \(C=1\), the same construction allows \(K\) to be arbitrarily large, and hence no finite robustness is possible.
\end{proof}

For \(C>1\), the construction gives the quantitative dependence \(\frac{C}{C-1}\), which decreases from infinity to two as \(C\) increases from one to two.
When \(C\ge2\), this term is at most two and is already covered by \(R\ge n\).
The interesting range for improving the prediction-dependent lower bound is therefore \(1<C<2\).

\subsection{A Lower Bound for Job-Wise Weighted Mechanisms}\label{sec:weighted-lower-bound}

The preceding lower bound applies without any restriction on how the jobs are allocated.
For job-wise weighted mechanisms, consistency can force many distinct jobs to satisfy the same large lower bound on the weight ratio between a common pair of machines.
The following lemma converts these repeated ratios directly into a worst-case instance.

\begin{lemma}\label{lem:repeated-ratios}
Let \(\mathbf r\) be a positive weight matrix.
Suppose that two distinct machines \(s,h\), a set \(J\) of \(n-1\) distinct jobs, and a value \(K\ge1\) satisfy \(\frac{r_{sj}}{r_{hj}}\ge K\) for every \(j\in J\).
If there is an additional job \(a\notin J\), then the worst-case approximation ratio of \(M_{\mathbf r}\) is at least \(n-1+K\).
\end{lemma}

\begin{proof}
For every machine \(i\), let
\(
    g_i=\left(\prod_{j\in J}r_{ij}\right)^{1/(n-1)}.
\)
The repeated-ratio assumption gives \(g_s/g_h\ge K\).
Fix a machine \(i\).
For each bijection \(\pi:N\setminus\{i\}\to J\), form the partial matching
\[
    Q_\pi=\{(i,a)\}\cup\{(\ell,\pi(\ell)):\ell\ne i\}.
\]
Averaging its weight relative to machine \(i\) over all \((n-1)!\) bijections, and then applying the arithmetic--geometric mean (AM--GM) inequality for each \(\ell\ne i\), gives
\begin{align*}
    \mathbb E_\pi\left[\sum_{(\ell,j)\in Q_\pi}
        \frac{r_{\ell j}}{r_{ij}}\right]
    &=1+\sum_{\ell\ne i}\frac1{n-1}
        \sum_{j\in J}\frac{r_{\ell j}}{r_{ij}}\\
    &\ge1+\sum_{\ell\ne i}
        \left(\prod_{j\in J}\frac{r_{\ell j}}{r_{ij}}\right)^{1/(n-1)}\\
    &=1+\sum_{\ell\ne i}\frac{g_\ell}{g_i}
     =\frac{\sum_{\ell\in N}g_\ell}{g_i}.
\end{align*}
Choose \(i\) minimizing \(g_i\).
Every ratio \(g_\ell/g_i\) is at least one, and
\(g_s/g_i\ge g_s/g_h\ge K\).
The last sum is therefore at least \(n-1+K\).
Hence some bijection \(\pi\) satisfies \(\sum_{(\ell,j)\in Q_\pi}r_{\ell j}/r_{ij}\ge n-1+K\).
Fix such a bijection \(\pi\).

It remains to realize this sum as an approximation ratio.
Choose \(\varepsilon>0\) smaller than every ratio \(r_{\ell j}/r_{ij}\) appearing in \(Q_\pi\).
For each \((\ell,j)\in Q_\pi\) with \(\ell\ne i\), set
\[
    p_{\ell j}=1,
    \qquad
    p_{ij}=\frac{r_{\ell j}}{r_{ij}}-\varepsilon,
\]
and make the processing time of \(j\) on every other machine sufficiently large.
Then \(r_{ij}p_{ij}<r_{\ell j}p_{\ell j}\), so \(M_{\mathbf r}\) assigns \(j\) to \(i\).
For the edge \((i,a)\), set \(p_{ia}=1\) and make the processing time of \(a\) on every other machine sufficiently large.
Give every remaining job processing time \(\delta>0\) on every machine.
The weighted mechanism puts load at least \(\sum_{(\ell,j)\in Q_\pi}r_{\ell j}/r_{ij}-n\varepsilon\) on \(i\), whereas assigning each job in \(Q_\pi\) to the machine paired with it gives a feasible schedule of makespan at most \(1+m\delta\), and hence \(\OPT(\mathbf p)\le1+m\delta\).
Letting \(\delta\) and then \(\varepsilon\) tend to zero proves the claim.
\end{proof}

\begin{theorem}\label{thm:weighted-lower-bound}
Let a prediction-dependent job-wise weighted mechanism be \(C\)-consistent and \(R\)-robust on \(n\ge2\) unrelated machines.
If \(C>1\), then
\(
    R\ge
    \max\left\{n,\frac{(n-1)C}{C-1}\right\}.
\)
If \(C=1\), then its robustness is unbounded.
\end{theorem}

\begin{proof}
When \(C\ge n\), the second term in the claimed bound is at most \(n\), so the bound reduces to \(R\ge n\). This follows as in Theorem~\ref{thm:basic-lower-bound}, by fixing the prediction and applying the lower bound for deterministic strategyproof scheduling.
When \(C=1\), the unboundedness follows from the same theorem. It therefore suffices to consider \(1<C<n\).

Choose a value $c\in\left(\frac{C-1}{n-1},1\right)$.
Relabel the \(n\) machines as one root machine \(0\) and leaf machines \(1,\ldots,n-1\). There is an anchor job \(a\) and, for each leaf \(i\), there are \(k\) regular jobs \(j_{i,1},\ldots,j_{i,k}\).
For a finite value \(H>C\), define the prediction by
\begin{align*}
    &\widehat p_{0a}=1,
    \qquad \widehat p_{ia}=H &&(i\ne0),\\
    &\widehat p_{0j_{i,t}}=\frac{c}{k},
    \qquad \widehat p_{ij_{i,t}}=\frac1k,
    \qquad \widehat p_{\ell j_{i,t}}=H
    &&(\ell\notin\{0,i\}).
\end{align*}
Assigning the anchor to the root and each regular job to its own leaf has makespan one.
Hence, \(\OPT(\widehat{\mathbf p})=1\).

When the prediction is correct, \(C\)-consistency rules out every \(H\)-entry.
Thus the anchor goes to the root, while a regular job goes either to the root or to its own leaf.
If the root receives \(q\) regular jobs, its load is \(1+qc/k\), and hence
\(
    q\le\frac{(C-1)k}{c}.
\)
At least \(k(n-1-(C-1)/c)\) regular jobs therefore remain on their leaves.
Choose a sufficiently large but finite \(k\) so that
\begin{align}\label{eq:replication-count}
    k\left(n-1-\frac{C-1}{c}\right)>(n-1)(n-2).
\end{align}
By the pigeonhole principle, some leaf \(h\) receives at least \(n-1\) of its own regular jobs.

Fix the weight matrix \(\mathbf r\) selected for this prediction.
For every such job \(j\) assigned to \(h\), the weighted-score comparison with the root gives
\(
    r_{hj}\frac1k\le r_{0j}\frac ck,
    \) and \(
    \frac{r_{0j}}{r_{hj}}\ge\frac1c.
\)
Choose \(n-1\) of these distinct jobs and use the anchor as the additional job in Lemma~\ref{lem:repeated-ratios}, with \(s=0\) and \(K=1/c\).
Lemma~\ref{lem:repeated-ratios} therefore yields
\(
    R\ge n-1+\frac1c.
\)
Letting \(c\) decrease to \((C-1)/(n-1)\), and choosing a finite \(k\) satisfying~\eqref{eq:replication-count} each time, proves
\(
    R\ge n-1+\frac{n-1}{C-1}.
\)
\end{proof}

The lower bound in Theorem~\ref{thm:weighted-lower-bound} matches the upper bound of \textsc{EdgeSkip} with an optimal reference schedule: setting \(\kappa=1\) in Theorem~\ref{thm:edge-skip} gives \(R_{1}(C)=\max\{n,\frac{(n-1)C}{C-1}\}\).
Therefore, over job-wise weighted mechanisms that work uniformly for every finite number of jobs, the optimal robustness of a \(C\)-consistent mechanism is
\(
\max\left\{n,\frac{(n-1)C}{C-1}\right\}
\)
for every \(C>1\).

\section{Experiments}\label{sec:experiments}

We evaluate the polynomial-time mechanisms under truthful reports on the HCSP
benchmark instances of Braun et al.~\cite{braun2001comparison}, using predictions obtained by synthetically perturbing the true processing-time matrices. The sole
performance measure is the ratio
\begin{align*}
    \rho(\mathbf p,\widehat{\mathbf p})
    =
    \frac{\MS(\mathbf p,\mathbf x(\mathbf p;\widehat{\mathbf p}))}
         {\OPT(\mathbf p)},
\end{align*}
where smaller values are better. All averages below are averages of \(\rho\).

\subsection{Setup}\label{sec:experimental-setup}

We compare \textsc{EdgeSkip} and \textsc{TolerantEdgeSkip} with
\textsc{ScaledGreedy}, \textsc{ErrorTolerantScaledGreedy}, and
\textsc{SimpleScaledGreedy}~\cite{balkanski2023strategyproof}.
For each job, \textsc{SimpleScaledGreedy} scales every machine's reported
processing time by the predicted speedup of that machine relative to the reference machine,
truncated to the interval \([1,n]\), and selects a machine with minimum weighted reported processing time. We also include
\textsc{Greedy}, which sends each job to a fastest machine under the reported
processing times, and the report-independent schedule computed by the
Lenstra--Shmoys--Tardos (LST) \(2\)-approximation algorithm from \(\widehat{\mathbf p}\),
denoted by \textsc{PredictionOnly}~\cite{lenstra1990approximation}. The five
prediction-dependent job-wise weighted mechanisms use this same LST reference
schedule; \textsc{PredictionOnly} returns the schedule itself. We set \(\gamma=1\) and
\(\bar\eta=2\) throughout and use fixed deterministic tie-breaking rules.

We use all 12 HCSP expected-time-to-compute matrices,
each with 512 jobs and 16 machines. For each of three fixed random seeds used
to construct instances, we select eight machines from every matrix and partition all 512 jobs
into 16 disjoint batches of 32 jobs.  For each resulting instance, we generate
\begin{align*}
    \widehat p_{ij}=p_{ij}\eta^{z_{ij}},
    \qquad
    \eta\in\{1,1.1,1.25,1.5,2,3,5,10\}.
\end{align*}
Each exponent matrix is rescaled so that \(\max_{i,j}|z_{ij}|=1\), making its
realized symmetric multiplicative error exactly \(\eta\). We use three error
structures: entrywise independent and identically distributed (i.i.d.) exponents
uniform on \([-1,1]\); machine--job error with correlated exponents
\(z_{ij}=a_i+b_j+0.15\epsilon_{ij}\), where \(a_i\), \(b_j\), and
\(\epsilon_{ij}\) are independent standard normal random variables representing
machine, job, and residual effects, respectively; and sparse error, in which only \(5\%\) of the entries
have nonzero exponents drawn uniformly from \([-1,1]\). For every \(\eta>1\), error
structure, and batch, we use three fixed perturbation seeds. The case
\(\eta=1\) is the common perfect prediction.

For evaluation only, we use a mixed-integer program to certify
\(\OPT(\mathbf p)\).  This value is used solely in the denominator of \(\rho\);
it is never supplied to a mechanism and is never used to construct the
reference schedule. For each of the seven reported methods, we use the
same 36,864 HCSP trial inputs. We first average over perturbation seeds, then
over batches, and then over the three constructions. Finally, we give equal
weight to each of the 12 source matrices; we call the resulting quantity the
macro average.

\subsection{Experimental Results}\label{sec:experimental-results}

Figure~\ref{fig:hcsp-headline} presents the direct HCSP comparisons
between each proposed mechanism and its \textsc{ScaledGreedy} counterpart.

\begin{figure}[H]
    \centering
    \includegraphics[width=\linewidth]{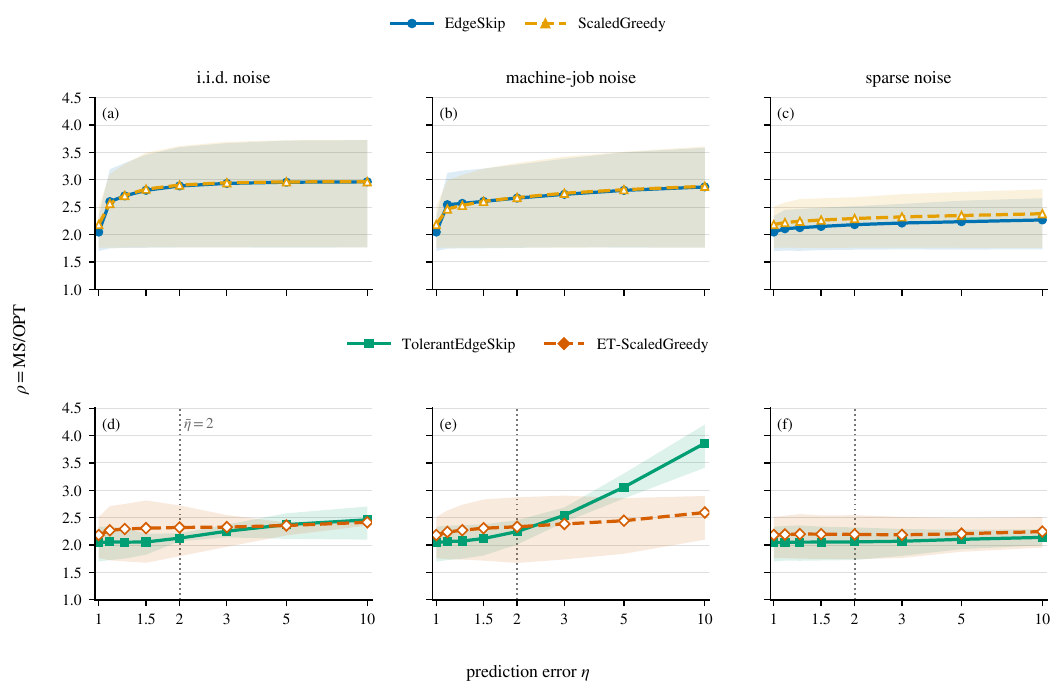}
    \caption{Direct comparisons on HCSP\@. Columns show i.i.d., machine--job,
    and sparse prediction errors.  The top row compares \textsc{EdgeSkip}
    with \textsc{ScaledGreedy}; the bottom row compares \textsc{TolerantEdgeSkip}
    with \textsc{ErrorTolerantScaledGreedy}.  Lines are equal-weight means over
    the 12 source matrices after the within-matrix averaging described in
    Section~\ref{sec:experimental-setup}; bands span the 25th--75th percentiles
    of those 12 matrix means.  The common perfect-prediction point is repeated
    in all three columns, and the dotted line in the lower row marks
    \(\bar\eta=2\). Lower is better.}
    \label{fig:hcsp-headline}
\end{figure}

Across the 21 noisy settings, \textsc{EdgeSkip} has lower
macro-averaged \(\rho\) than \textsc{ScaledGreedy} in 17, while
\textsc{TolerantEdgeSkip} improves on \textsc{ErrorTolerantScaledGreedy} in
16. Both paired comparisons favor our mechanisms at all seven noisy values of \(\eta\) under
sparse error. The advantage of our
tolerant mechanism is not uniform: under machine--job error at \(\eta=10\),
\textsc{TolerantEdgeSkip} reaches \(3.8601\), compared with \(2.5948\) for
\textsc{ErrorTolerantScaledGreedy}, and is the worst of all seven methods in
that setting.

With perfect HCSP predictions, \textsc{EdgeSkip} reduces the
macro-averaged \(\rho\) from \(2.1837\) to \(2.0457\) relative to
\textsc{ScaledGreedy}; the same values hold for their tolerant variants.
\textsc{SimpleScaledGreedy} performs best at \(\eta=1\), with \(1.3757\), followed
by \textsc{PredictionOnly} at \(1.3854\).

The complete HCSP comparison is reported in
Appendix~\ref{app:experiments}.

\section{Conclusion}\label{sec:conclusion}

We establish tighter consistency--robustness tradeoffs for strategyproof scheduling with predictions.
For job-wise weighted mechanisms, our upper and lower bounds match in the absence of computational restrictions.
With an optimal reference schedule, \textsc{EdgeSkip} is \(C\)-consistent and attains the optimal robustness
\(\max\{n,\frac{(n-1)C}{C-1}\}\) for every \(C>1\).
With a \(2\)-approximate reference schedule and \(\gamma=1\), it is \(4\)-consistent and \((2n-2)\)-robust, improving the polynomial-time tradeoff.
The same improvement extends to the error-tolerant setting.
Thus, we determine the exact weighted frontier and improve its polynomial-time guarantees, while a gap remains for unrestricted deterministic strategyproof mechanisms.

The main open problem is whether a general deterministic strategyproof mechanism can strictly improve on the optimal weighted frontier when \(n\ge3\), or whether stronger unrestricted lower bounds preclude such an improvement.
Another question is to determine the optimal polynomial-time frontier, since an optimal predicted reference schedule cannot in general be computed efficiently.
Further directions include obtaining the same tradeoffs using less prediction information and understanding how randomization or more refined error measures change the frontier.

\bibliographystyle{plain}
\bibliography{mybibfile}

\appendix

\section{The Two-Machine Case}\label{app:two-machines}

This appendix isolates the special case of two unrelated machines. Its purpose is to record a simple threshold mechanism whose proof avoids the forest argument used for \textsc{EdgeSkip} and whose guarantee can be stronger when the reference schedule is only approximately optimal.

Suppose \(N=\{1,2\}\), and fix parameters \(C>1\) and \(\kappa\ge1\).
Let \(\widehat\sigma\) be a \(\kappa\)-approximate reference schedule computed from the prediction \(\widehat{\mathbf p}\), and set \(\lambda_C=\max\{1,\frac{1}{C-1}\}\).
For each job, the mechanism favors its reference machine by a factor of \(\lambda_C\): the reference machine receives the job unless the other machine is more than \(\lambda_C\) times faster.

\begin{algorithm}[H]
\caption{Two-machine threshold mechanism}\label{alg:two-machine-threshold}
\KwIn{Prediction \(\widehat{\mathbf p}\), reports \(\mathbf p\), parameters \(C>1\) and \(\kappa\ge1\)}
Compute a fixed \(\kappa\)-approximate reference schedule \(\widehat\sigma\) for \(\widehat{\mathbf p}\), and set \(\lambda_C\leftarrow\max\{1,\frac{1}{C-1}\}\)\;
\For{each job \(j\in M\)}{
    Assign \(j\) to \(\widehat\sigma(j)\) temporarily\;
    \For{each machine \(k\in N\)}{
        \If{\(p_{\widehat\sigma(j),j}>\lambda_C p_{kj}\)}{
            Reassign \(j\) to \(k\)\;
        }
    }
}
\end{algorithm}

The rule is a job-wise weighted mechanism: job \(j\) has weight \(1\) on \(\widehat\sigma(j)\) and weight \(\lambda_C\) on every \(k\in N\setminus\{\widehat\sigma(j)\}\).
The reference machine wins a tie.

\begin{theorem}\label{thm:two-machine-upper}
For every \(C>1\) and \(\kappa\ge1\), the two-machine threshold mechanism is deterministic, strategyproof, \(\kappa C\)-consistent, and \(\max\{2,\frac{C}{C-1}\}\)-robust.
\end{theorem}

\begin{proof}
The reference schedule and \(\lambda_C\) depend only on the public prediction and on the public parameters \(C\) and \(\kappa\).
Together with the fixed tie-breaking rule, these choices make the allocation deterministic.
Since the weights are fixed independently of the reports and each job is allocated independently, each per-job allocation is monotone in the reported processing time. Summing the resulting per-job inequalities gives weak monotonicity, so the allocation rule admits strategyproof payments.

We next prove consistency.
Suppose \(\mathbf p=\widehat{\mathbf p}\), write \(\mathbf x\) for the resulting allocation, and let \(T\) be the makespan of \(\widehat\sigma\) on \(\mathbf p\).
Since \(\widehat\sigma\) is \(\kappa\)-approximate, \(T\le\kappa\OPT(\mathbf p)\).
Fix a machine \(h\), and let \(k\) be the other machine.
The jobs that both the reference schedule and the mechanism assign to \(h\) have total processing time at most \(T\).
Every remaining job assigned to \(h\) has reference machine \(k\).
Such a job moves to \(h\) only if \(p_{kj}>\lambda_C p_{hj}\), and hence \(p_{hj}<\frac{p_{kj}}{\lambda_C}\).
These jobs form a subset of the reference load of \(k\), whose total is at most \(T\).
Consequently,
\begin{align*}
    L_h(\mathbf p,\mathbf x)
    \le T+\frac{T}{\lambda_C}
    \le C T
    \le \kappa C\OPT(\mathbf p).
\end{align*}
The inequality \(T+T/\lambda_C\le CT\) is an equality when \(1<C\le2\); when \(C>2\), it follows from \(\lambda_C=1\) and \(2<C\).
The bound holds for both machines, proving \(\kappa C\)-consistency.

For robustness, fix an arbitrary pair \((\mathbf p,\widehat{\mathbf p})\), let \(\mathbf x\) be the resulting allocation, and fix a machine \(h\).
Again let \(k\) be the other machine.
The jobs assigned to \(h\) by both \(\mathbf x\) and \(\mathbf x^*(\mathbf p)\) contribute at most \(\OPT(\mathbf p)\).
Consider instead a job \(j\) assigned to \(h\) by \(\mathbf x\) and to \(k\) by \(\mathbf x^*(\mathbf p)\).
If \(h\) is its reference machine, the threshold rule gives \(p_{hj}\le\lambda_C p_{kj}\).
If \(k\) is its reference machine, the job moves to \(h\), so \(p_{hj}<\frac{p_{kj}}{\lambda_C}\le\lambda_C p_{kj}\), where the last inequality uses \(\lambda_C\ge1\).
Thus every such job satisfies the same upper bound, and their total processing time on \(k\) under \(\mathbf x^*(\mathbf p)\) is at most \(\OPT(\mathbf p)\).
It follows that
\begin{align*}
    L_h(\mathbf p,\mathbf x)
    \le (1+\lambda_C)\OPT(\mathbf p).
\end{align*}
Finally, \(1+\lambda_C=\max\{2,\frac{C}{C-1}\}\).
Taking the maximum over the two machines proves the robustness claim.
\end{proof}

To compare with Theorem~\ref{thm:edge-skip}, specialize that theorem to \(n=2\) and set \(\gamma=C-1\). The consistency parameter in that theorem is then \((1+\gamma)\kappa=\kappa C\), so both mechanisms are \(\kappa C\)-consistent. If \(\kappa=1\), Theorem~\ref{thm:edge-skip} gives
\[
    R_{1}(C)
    =\max\left\{2,\frac{C}{C-1}\right\},
\]
so the two robustness guarantees coincide for every \(C>1\). For arbitrary \(\kappa\), they also coincide when \(C\ge2\), in which case both robustness bounds equal two.

When \(\kappa>1\) and \(1<C<2\), the threshold mechanism is strictly stronger: its robustness is \(C/(C-1)\), whereas the specialization of Theorem~\ref{thm:edge-skip} is
\[
    R_{\kappa}(\kappa C)
    =
    \begin{cases}
        \displaystyle 1+\frac{2-1/\kappa}{C-1},
        &1<C\le1+1/\kappa,\\[2mm]
        \displaystyle \frac{2}{C-1},
        &1+1/\kappa<C<2.
    \end{cases}
\]
Thus the appendix improves the two-machine guarantee in precisely this approximate-reference regime, while Theorem~\ref{thm:edge-skip} applies to an arbitrary number of machines.

\section{Additional Experimental Results}\label{app:experiments}

Figure~\ref{fig:hcsp-all} gives the complete seven-method HCSP comparison, and
Table~\ref{tab:hcsp-complete} reports every corresponding macro average. These
results use the same hierarchy and the same sole performance measure \(\rho\)
as the main experiments.

\begin{figure}[H]
    \centering
    \includegraphics[width=\linewidth]{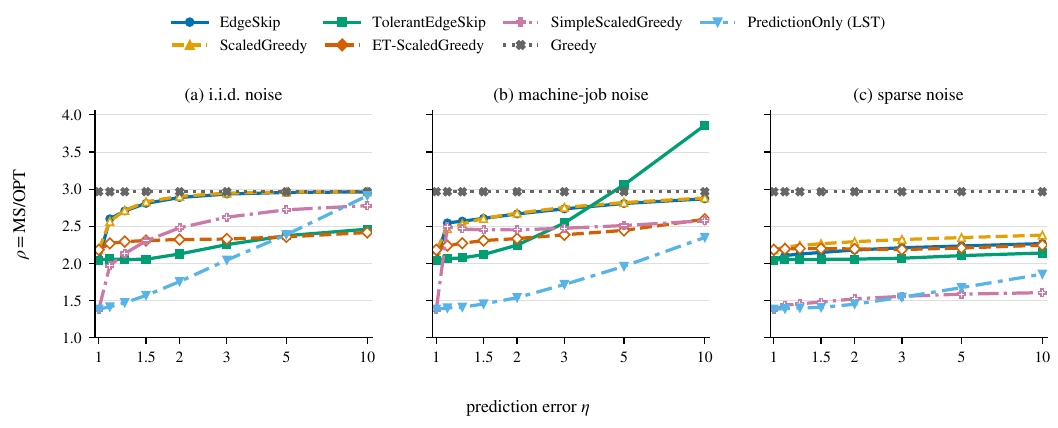}
    \caption{All seven methods on HCSP under i.i.d., machine--job, and sparse
    prediction errors. Each point is the equal-weight mean over the 12 source
    matrices after within-matrix averaging. The common perfect-prediction
    value is repeated at \(\eta=1\) in all three panels. Lower is better.}
    \label{fig:hcsp-all}
\end{figure}

\begin{table}[H]
  \centering
  \caption{Complete results on the classic HCSP matrices. Entries are source-matrix macro-averaged $\rho=\MS/\OPT$; lower is better. ES, SG, TES, ET-SG, SSG, and Pred.-LST denote \textsc{EdgeSkip}, \textsc{ScaledGreedy}, \textsc{TolerantEdgeSkip}, \textsc{ErrorTolerantScaledGreedy}, \textsc{SimpleScaledGreedy}, and \textsc{PredictionOnly} (LST), respectively. The perfect-prediction row is shared by all three noise families at $\eta=1$.}
  \label{tab:hcsp-complete}
  \small
  \setlength{\tabcolsep}{3.6pt}
  \resizebox{\columnwidth}{!}{%
  \begin{tabular}{llrrrrrrr}
    \toprule
    Prediction & $\eta$ & ES & SG & TES & ET-SG & SSG & Greedy & Pred.-LST \\
    \midrule
    Perfect & 1 & 2.0457 & 2.1837 & 2.0457 & 2.1837 & \textbf{1.3757} & 2.9652 & 1.3854 \\
    \addlinespace[2pt]
    i.i.d. & 1.1 & 2.6010 & 2.5667 & 2.0595 & 2.2693 & 1.9708 & 2.9652 & \textbf{1.4156} \\
     & 1.25 & 2.7063 & 2.7145 & 2.0529 & 2.2934 & 2.1334 & 2.9652 & \textbf{1.4710} \\
     & 1.5 & 2.8087 & 2.8301 & 2.0565 & 2.3084 & 2.3025 & 2.9652 & \textbf{1.5689} \\
     & 2 & 2.8875 & 2.9037 & 2.1271 & 2.3218 & 2.4790 & 2.9652 & \textbf{1.7553} \\
     & 3 & 2.9346 & 2.9471 & 2.2534 & 2.3283 & 2.6227 & 2.9652 & \textbf{2.0455} \\
     & 5 & 2.9559 & 2.9614 & 2.3735 & \textbf{2.3575} & 2.7225 & 2.9652 & 2.3878 \\
     & 10 & 2.9628 & 2.9642 & 2.4609 & \textbf{2.4165} & 2.7783 & 2.9652 & 2.9108 \\
    \addlinespace[2pt]
    Machine-job & 1.1 & 2.5456 & 2.4682 & 2.0650 & 2.2393 & 2.4941 & 2.9652 & \textbf{1.3999} \\
     & 1.25 & 2.5688 & 2.5326 & 2.0743 & 2.2716 & 2.4645 & 2.9652 & \textbf{1.4107} \\
     & 1.5 & 2.6065 & 2.6064 & 2.1196 & 2.3082 & 2.4522 & 2.9652 & \textbf{1.4547} \\
     & 2 & 2.6663 & 2.6769 & 2.2465 & 2.3342 & 2.4519 & 2.9652 & \textbf{1.5386} \\
     & 3 & 2.7348 & 2.7547 & 2.5442 & 2.3859 & 2.4736 & 2.9652 & \textbf{1.7165} \\
     & 5 & 2.8073 & 2.8199 & 3.0576 & 2.4450 & 2.5103 & 2.9652 & \textbf{1.9613} \\
     & 10 & 2.8699 & 2.8814 & 3.8601 & 2.5948 & 2.5736 & 2.9652 & \textbf{2.3500} \\
    \addlinespace[2pt]
    Sparse & 1.1 & 2.1062 & 2.2191 & 2.0513 & 2.1939 & 1.4367 & 2.9652 & \textbf{1.3902} \\
     & 1.25 & 2.1263 & 2.2435 & 2.0540 & 2.2011 & 1.4569 & 2.9652 & \textbf{1.4015} \\
     & 1.5 & 2.1481 & 2.2651 & 2.0545 & 2.1998 & 1.4799 & 2.9652 & \textbf{1.4079} \\
     & 2 & 2.1810 & 2.2931 & 2.0583 & 2.1949 & 1.5231 & 2.9652 & \textbf{1.4514} \\
     & 3 & 2.2109 & 2.3225 & 2.0702 & 2.1872 & 1.5569 & 2.9652 & \textbf{1.5449} \\
     & 5 & 2.2347 & 2.3474 & 2.1052 & 2.2078 & \textbf{1.5864} & 2.9652 & 1.6744 \\
     & 10 & 2.2670 & 2.3809 & 2.1408 & 2.2460 & \textbf{1.6070} & 2.9652 & 1.8562 \\
    \bottomrule
  \end{tabular}%
  }
\end{table}

\end{document}